\documentclass[final,conference]{IEEEtran}

\usepackage[final]{graphicx}
\usepackage[reqno]{amsmath}
\usepackage{amssymb}
\usepackage{url}
\usepackage{color}
\usepackage{times}
\usepackage{wrapfig}
\usepackage{subfig}
\usepackage{amsthm}

\newfont{\bbb}{msbm10 scaled 500}

\newfont{\bb}{msbm10 scaled 1100}

\newcommand{\EE}{\mbox{\bb E}}

\newcommand{\Prob}{\textrm{Pr}}

\theoremstyle{plain}
\newtheorem{theorem}{Theorem}
\newtheorem{lemma}[theorem]{Lemma}

\definecolor{OXO-emph}{RGB}{153,0,0}

\title{Lossy Compression of Exponential and Laplacian Sources using Expansion Coding}

\author{\IEEEauthorblockN{Hongbo~Si, O.~Ozan~Koyluoglu,
and Sriram~Vishwanath}
\IEEEauthorblockA{Laboratory for Informatics, Networks, and Communications\\
Wireless Networking and Communications Group\\
The University of Texas at Austin\\
1 University Station, C0806, Austin, TX 78712\\
Email: \{sihongbo,ozan\}@mail.utexas.edu, sriram@austin.utexas.edu}}

\begin{document}

\maketitle


\begin{abstract}

A general method of source coding over expansion is proposed in this paper, which enables one to reduce the problem of compressing an analog (continuous-valued source) to a set of much simpler problems, compressing discrete sources. Specifically, the focus is on lossy compression of exponential and Laplacian sources, which is subsequently expanded using a finite alphabet prior to being quantized. Due to decomposability property of such sources, the resulting random variables post expansion are independent and discrete. Thus, each of the expanded levels corresponds to an independent discrete source coding problem, and the original problem is reduced to coding over these parallel sources with a total distortion constraint. Any feasible solution to the  optimization problem is   an achievable rate distortion pair of the original  continuous-valued source compression problem. Although finding the solution to this optimization problem at every distortion is hard, we show that our expansion coding scheme presents a good solution in the low distrotion regime. Further, by adopting low-complexity codes designed for discrete source coding, the total coding complexity can be tractable in practice.

\end{abstract}


\section{Introduction}

The compression of continuous-valued sources remains one of the most well-studied (and practically valuable) research directions in Information Theory. Given the increased importance of voice, video and other multimedia, all of which are typically "analog" in nature, the value associated with low-complexity algorithms to compress continuous-valued data is likely to remain significant in the years to come.

For discrete-valued "finite alphabet", both the associated coding theorem \cite{Cover:IT1991} and practically-meaningful coding schemes are now well known. Trellis based quantizers \cite{Viterbi:Trellis74} are the first to achieve the rate distortion tradeoff, but with encoding complexity scaling exponentially with the constraint length. Later, Matsunaga and Yamamoto \cite{Matsunaga:LDPC2003} show that  a low density parity check (LDPC) ensemble, under suitable conditions on ensemble structure, can achieve the rate distortion bound using an optimal decoder. Further, \cite{Matsunaga:LDGM2010} shows that low density generator matrix (LDGM) codes, as the dual of LDPC codes, with suitably irregular degree distributions, empirically perform close to the Shannon rate-distortion bound with message-passing algorithms. More recently, polar codes \cite{Arikan:Channel08}, are the first provably rate distortion limit achievable codes with low encoding and decoding complexity \cite{Korada:Source10}.

In the case of analog sources, although both practical coding schemes as well as theoretical analysis is very heavily studied, a very limited literature exists that connects theory with low-complexity codes in practice. The most relevant literature in this context is on lattice compression and its low-density constructions \cite{Zamir:Lattice09}, although this literature is somewhat limited in scope and application.

 In the domains of image compression and speech coding, Laplacian and exponential distributions are widely adopted as natural models of correlation between pixels and amplitude of voice \cite{Gallager:Information68}. Exponential distribution is also fundamental in characterizing continuous-time Markov processes \cite{Verdu:Exponential96}. Although the rate distortion functions for both have  been known for decades, there is still a gap between theory and existing low-complexity coding schemes for them. Some schemes have been proposed, primary for the medium to high distortion regime, such as Markov chain Monte Carlo (MCMC) based approach \cite{Baron:MCMC12}. Our general understanding of low-complexity coding schemes, particular for the low-distortion regime, remains limited.

In this paper, we present an expansion coding scheme for both exponential and Laplacian sources, which not only  performs well  in the low distortion regime, but also can be implemented with low encoding and decoding complexity. Previously, our work in \cite{Ozan:Expansion12} considers the dual problem of expansion coding for the channel coding case, where exponential noise channels are converted to coding over a set of parallel (and independent) discrete channels. Further, adopting capacity achieving codes to the resulting parallel channels, expansion coding is shown to achieve the channel capacity at high SNR with low complexity. For source coding, we utilize a similar approach here. Consider expanding exponential and Laplace sources into binary sequences, and coding over the resulting set of parallel discrete sources. By carefully choosing the parameters for each of the parallel lossy compression problems, we show that the achievable rates for original source approaches the rate distortion limit, in ratio, in the low distortion regime.

The rest of paper is organized as follows. The next
section describes the background of source coding problem. In Section III and VI, we
present the main results of this paper, expansion coding technique for exponential and Laplacian source, respectively.
The paper concludes with a discussion section.


\section{Background}

\subsection{Source Coding Problem}

Consider an i.i.d. source $X_1, X_2,\ldots, X_n$. A $(2^{nR},n)$-rate distortion code consists of an encoding function $g:\mathbb{R}^n\to\mathcal{M}$, where $\mathcal{M}\triangleq\{1,\ldots,2^{nR}\}$, and a decoding function $h:\mathcal{M}\to\mathbb{R}^n$, which codes $X^n$ to an estimate $\hat{X}^n$. Then, a rate and distortion pair $(R,D)$ is said to be achievable if there exists a sequence of $(2^{nR},n)$-rate distortion codes with $\lim\limits_{n\to\infty}\mathbb{E}[d(X^n,\hat{X}^n)]\leq D$. The rate distortion function $R(D)$ is the infimum of such rates, and by Shannon's theorem \cite{Cover:IT1991}, we have:
\begin{equation}
R(D)=\min_{f(\hat{x}|x):\mathbb{E}[d(X^n,\hat{X}^n)]\leq D}I(X;\hat{X}).\label{fun:rate_distortion_theorem}
\end{equation}

\subsection{Decomposability of Exponential Distribution}

The intuition underlying expansion coding originates from decomposability property  of exponential random variables, where it can be expressed as summation of a set of independent  discrete-valued random variables.  The following lemma crystallizes this concept:

\begin{lemma}[\cite{Ozan:Expansion12}, \cite{Marsaglia:Random71}]\label{lem:exponential_expansion}
Let $B_l$'s be independent Bernoulli random variables, and their distribution is given by parameters
$b_{l}\triangleq\Prob\{B_l=1\}$.
Then, the random variable
$$B=\sum_{l=-\infty}^{\infty} 2^l B_l$$
 is exponentially distributed
with mean $\lambda^{-1}$,
if and only if the choice of $b_{l}$ is given by
\begin{equation}
b_l=\frac{1}{1+e^{\lambda 2^l}}. \nonumber
\end{equation}
\end{lemma}
\begin{proof}
Proof is given in \cite{Ozan:Expansion12} and \cite{Marsaglia:Random71}, and follows from the memoryless property of exponential distribution.
For completeness, we provide the proof in Appendix~\ref{sec:AppA}.
\end{proof}

A set of typical numerical values of $b_l$s by fixing $\lambda=1$ is shown in Fig.~\ref{fig:Exp_Prob}. It is evident that $b_l$ approaches  $0$
for the ``higher" levels and approaches $0.5$ for what we refer to as  ``lower" levels. Hence, the primary non-trivial levels within which coding is meaningful are the so-called ``middle" ones, which provides the basis for  truncating the number of levels to a finite value without a significant loss in performance.

\begin{figure}[h!]
 \centering
 \includegraphics[width=0.9\columnwidth]{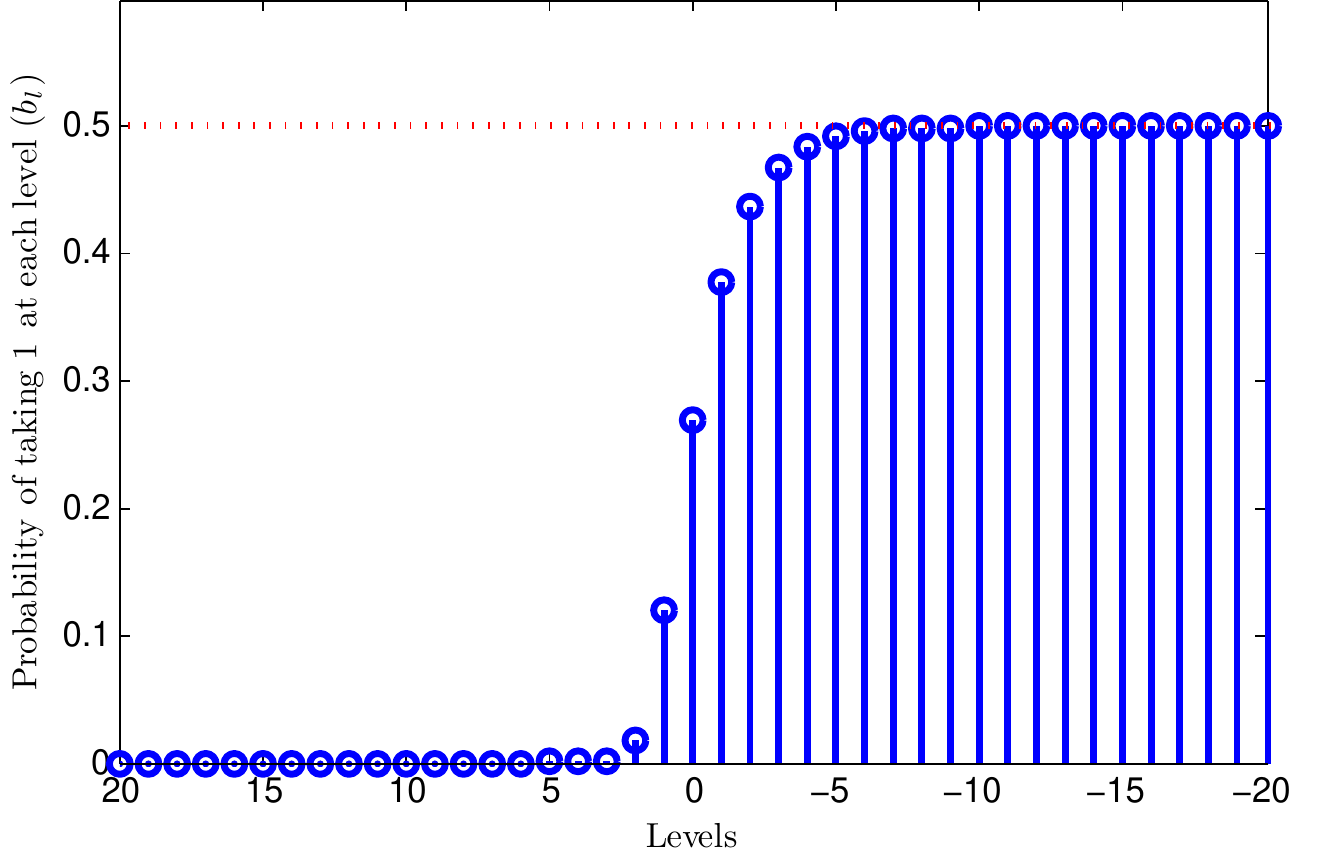}
 \caption{{\bf Numerical results for a set of $b_l$ for different levels with $\lambda=1$.}
}
\label{fig:Exp_Prob}
\end{figure}


\section{Exponential Source Coding}

\subsection{Problem Setup}

Consider an i.i.d. exponential source $X_1, X_2,\ldots, X_n$, i.e. omitting index $i$, the probability density function is given by
\begin{equation}
f_X(x)=\lambda e^{-\lambda x},\quad x\geq 0,\label{fun:exp_pdf}
\end{equation}
where $\lambda$ is the parameter of exponential distribution, i.e. $\mathbb{E}[X]=1/\lambda$. Distortion measure of concern is one-sided error distortion, i.e.
\begin{equation}
d(x^n,\hat{x}^n)=\left\{\begin{array}{ll}
\frac1n\sum\limits_{i=1}^n(x_i-\hat{x}_i),&\text{if } x_i\geq\hat{x}_i,\\
\infty,&\text{otherwise.}\end{array}\right.
\label{fun:exp_distortion_definition}
\end{equation}

This setup is equivalent to \cite{Verdu:Exponential96}, where another distortion measure is considered.

\begin{lemma}[\cite{Verdu:Exponential96}]\label{lem:exp_rate_distortion}
The rate distortion function for exponential source with one-sided error distortion is given by
\begin{align}
R(D)=\left\{\begin{array}{ll}
-\log  (\lambda D), &0\leq D\leq \frac{1}{\lambda},\\
0,&D>\frac{1}{\lambda}.
\end{array}
\right.
\end{align}
Moreover, the optimal conditional distribution to achieve the limit is given by
\begin{align}
f^*_{X|\hat{X}}(x|\hat{x})=\frac{1}{D} e^{- (x-\hat{x})/D},\quad x\geq\hat{x}\geq0.\label{fun:exp_optimal_conditional}
\end{align}
\end{lemma}
\begin{proof}
Proof is given in \cite{Verdu:Exponential96}, and it is based on the observation that among the ensemble of all probability density functions with positive support set and mean constraint, exponential distribution maximizes the differential entropy. By designing a test channel from $\hat{X}$ to $X$, with additive noise distributed as exponential with parameter $1/D$, both the infimum mutual information and optimal conditional distribution can be characterized. Details can be found in Appendix~\ref{sec:AppB}.
\end{proof}

\subsection{Expansion Coding}

Using Lemma \ref{lem:exponential_expansion}, we can reconstruct exponential distribution with parameter $\lambda$ by a set of discrete Bernoulli random variables. In particular, the expansion of exponential source over levels ranging from $-L_1$ to $L_2$ can be expressed as
\begin{equation}
X_i=\sum_{l=-L_1}^{L_2}2^lX_{i,l},\quad i=1,2,\ldots,n,\label{fun:exp_expansion_source}
\end{equation}
where $X_{i,l}$ are Bernoulli random variables with parameter
\begin{align}
p_{l}=\Prob\{X_{i,l}=1\}=\frac{1}{1+e^{\lambda 2^l}}.\label{fun:pl}
\end{align}
The expansion will perfectly approximate exponential source by letting $L_1,L_2\to\infty$. Consider a similar expansion of the source estimate, i.e.
\begin{equation}
\hat{X}_i=\sum_{l=-L_1}^{L_2}2^l \hat{X}_{i,l},\quad i=1,2,\ldots,n, \label{fun:exp_estimate_expansion}
\end{equation}
where $\hat{X}_{i,l}$ is resulting Bernoulli random variable with parameter $\hat{p}_l=\Prob\{\hat{X}_{i,l}=1\}$.

Using the concept of expansion, the original problem of coding for a continuous source can be translated to a problem of coding for a set of independent binary sources. In other words, the original optimization problem over all possible continuous densities has been converted to another one with finite parameters. This transformation, although seemingly obvious, is valuable as we already have powerful coding schemes over discrete sources achieving rate distortion limits with low complexity. In particular, we design two schemes for the binary source coding problem at each level.

\subsubsection{Coding with one-sided distortion}
We formulate each level as the binary source coding problem under the following one-sided distortion constraint:
\begin{equation}
d_O(x_l,\hat{x}_l)=\bold{1}_{\{x_l> \hat{x}_l\}}=\bold{1}_{\{x_l=1,\hat{x}=0\}}.
\end{equation}

Denoting the distortion at level $l$ as $d_l$, an asymmetric test channel (Z-channel) from $\hat{X}_{l}$ to $X_{l}$ can be constructed, where
\begin{align}
\Prob\{X_l=1|\hat{X}_l=0\}=\frac{d_l}{1-p_l+d_l}.\nonumber
\end{align}
Then, it is straightforward to get $p_l-\hat{p}_l=d_l$, and the achievable rate is given by
\begin{equation}
R_l=H(p_l)-(1-p_l+d_l)H\left(\frac{d_l}{1-p_l+d_l}\right).
\end{equation}

Due to the decomposability property as stated previously, the coding scheme provided will be over a set of parallel discrete levels indexed by $l=-L_1,\ldots,L_2$ correspondingly. Thus, by adopting rate distortion limit achieving codes over each level, our expansion coding scheme readily achieves the following result:
\begin{theorem}\label{thm:exp_rate_1}
For an exponential source, expansion coding achieves the rate distortion pair given by
\begin{align}
R^{(1)}&=\sum_{l=-L_1}^{L_2}R_l,\label{eqn:R1}\\
D^{(1)}&=\sum_{l=-L_1}^{L_2}2^ld_l+o( 2^{-L_2}/\lambda)+o(2^{-L_1}),\label{eqn:D1}
\end{align}
for any $L_1,L_2>0$, and $d_l\in[0,0.5]$ for $l\in\{-L_1,\cdots,L_2\}$, where $p_l$ is given by \eqref{fun:pl}.
\end{theorem}
\begin{proof}
See Appendix~\ref{sec:AppC}.
\end{proof}
Note that, the last two terms in \eqref{eqn:D1} are distortion resulting from the truncation (of the levels) and vanish in the limit of large number of levels. In later parts of this section, we characterize the number of levels required in order to bound the resulting distortion within a constant gap.

\subsubsection{Successive encoding and decoding} In the scheme above, we formulate each level as a Z-channel such that $X_l\geq \hat{X}_l$. However, it is not necessary to have this relationship to guarantee $X\geq \hat{X}$. To this end, we introduce successive coding scheme, where encoding and decoding start from the highest level $L_2$ to the lowest. At a certain level $l$, if all higher levels are decoded as $x_{k}=\hat{x}_{k}$ for $k>l$, then we must model level $l$ as binary source coding with a one-sided distortion (test channel is Z-channel). Otherwise, we formulate this level as binary source coding with symmetric distortion (test channel is binary symmetric channel). In particular for the later case, the distortion is Hamming distortion, i.e.
\begin{equation}
d_H(x_l,\hat{x}_l)=\bold{1}_{\{x_l\neq \hat{x}_l\}}.
\end{equation}

Denoting the equivalent distortion at level $l$ as $d_l$, i.e. $\mathbb{E}[X_l-\hat{X}_l]=d_l$, then the symmetric test channel from $\hat{X}_l$ to $X_l$ could be formulated as
\begin{equation}
\text{Pr}\{X_l=1|\hat{X}_l=0\}=\text{Pr}\{X_l=0|\hat{X}_l=1\}=\frac{d_l}{1-2p_l+2d_l}.\nonumber
\end{equation}
Hence, the achievable rate at level $l$ is given by
\begin{equation}
\bar{R}_l=H(p_l)-H\left(\frac{d_l}{1-2p_l+2d_l}\right).
\end{equation}

Based on these observations, we have the following achievable result:
\begin{theorem}\label{thm:exp_rate_2}
For exponential source, applying successive coding, expansion coding achieves the rate distortion pair given by
\begin{align}
R^{(2)}&=\sum_{l=-L_1}^{L_2}\left[q_l R_l +\left(1-q_l\right)\bar{R}_l\right],\label{eqn:R2}\\
D^{(2)}&=\sum_{l=-L_1}^{L_2}2^ld_l+o( 2^{-L_2}/\lambda)+o(2^{-L_1}),\label{eqn:D2}
\end{align}
for any $L_1,L_2>0$, and $d_l\in[0,0.5]$ for $l\in\{-L_1,\cdots,L_2\}$. Here, $p_l$ is given by \eqref{fun:pl},
and $q_l$ denotes the probability that all higher levels are encoded as equivalent and its value is given by
\begin{equation}
q_l= \prod_{k=l+1}^{L_2}(1-d_k).
\end{equation}
\end{theorem}
\begin{proof}
See Appendix~\ref{sec:AppD}.
\end{proof}

In this sense, the achievable pairs in Theorem~\ref{thm:exp_rate_1} and \ref{thm:exp_rate_2} are both given by optimization problems over a set of parameters $\{d_{-L_1},\ldots,d_{L_2}\}$. However, the problems are not convex, and effective theoretical analysis or numerical calculation cannot be adopted here for an optimal solution. But, by a heuristic choice of $d_l$, we can still get a good performance. Inspired from the fact that the optimal scheme models noise as exponential with parameter $1/D$ in test channel, we design
\begin{eqnarray}
d_l=\frac{1}{1+e^{2^l/D}}.\label{eqn:dl}
\end{eqnarray}

We note that higher levels get higher priority and
lower distortion with this choice, which is consistent with the intuition.
Then, the proposed expansion coding scheme provably approaches the rate distortion function for the whole distortion within a small constant gap.
\begin{theorem}\label{thm:exp_bound}
For any $D\in[0,1/\lambda]$, there exists a constant $c>0$, such that for $L_1,L_2>-\log (\lambda D)$, the achievable rate pairs obtained from expansion coding schemes are both within $c$ bit gap to Shannon rate distortion function, i.e.
$$R^{(1)}-R(D^{(1)})\leq c,$$
$$R^{(2)}-R(D^{(2)})\leq c,$$
where $D^{(1)}$ and $D^{(2)}$ are given by \eqref{eqn:D1} and \eqref{eqn:D2} respectively, with a choice of $d_l$ as in \eqref{eqn:dl}.
\end{theorem}
\begin{proof}
See Appendix~\ref{sec:AppE}.
\end{proof}

\subsection{Numerical Result}
Numerical results showing achievable rates along with the rate distortion limit are plotted in Fig.~\ref{fig:Exp_Rate}. It is evident that both forms of  expansion coding  perform within a constant gap of the limit. Theorem~\ref{thm:exp_bound} showcases that this gap is bounded by a constant. Here, numerical results show that the gap is not necessarily as wide as predicted by the analysis. Specially in the low distortion region, the gap is numerically found to correspond to 0.24 bit and 0.43 bit for each coding scheme respectively.

\begin{figure}[t]
 \centering
 \includegraphics[width=\columnwidth]{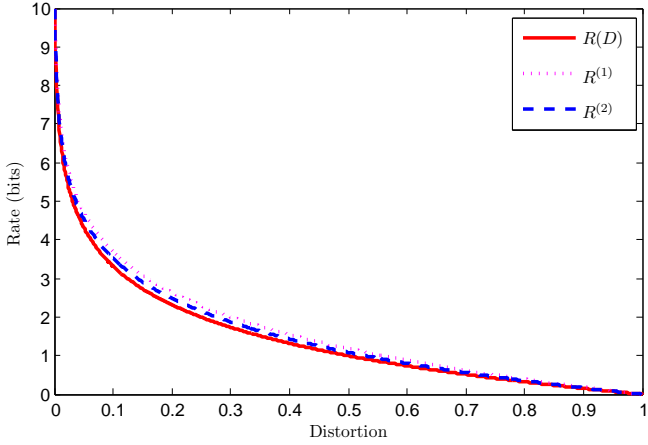}
 \caption{{\bf Achievable rate distortion pairs using expansion coding for exponential distribution with one-sided error distortion.} In this numerical result, we set $\lambda=1$. $R(D)$ (red-solid) is rate distortion limit; $R^{(1)}$ (purple-dotted) is the achievable rate given by Theorem~\ref{thm:exp_rate_1}; $R^{(2)}$ (blue-dashed) is the achievable rate given by Theorem~\ref{thm:exp_rate_2}.
}
\label{fig:Exp_Rate}
\end{figure}


\section{Laplacian Source Coding}

\subsection{Problem Setup}
In this section, we focus on Laplacian source coding.
Consider another i.i.d. exponential source $X_1,X_2,\ldots,X_n$, i.e. omitting index $i$, the probability density function is given by
\begin{equation}
f_X(x)=\frac{\lambda}{2}e^{-\lambda |x|},\quad x\in\mathbb{R},\label{fun:laplace_pdf}
\end{equation}
where $\lambda$ is the parameter of Laplace distribution, i.e. $\mathbb{E}[|X|]=1/\lambda$. Distortion measure here is absolute value error distortion, i.e.
\begin{equation}
d(x^n,\hat{x}^n)=\frac1n\sum_{i=1}^n|x_i-\hat{x}_i|.\label{fun:distortion_definition}
\end{equation}

\begin{lemma}[\cite{Cover:Laplace91}]\label{lem:laplace_rate_distortion}
The rate distortion function for Laplacian source with parameter $\lambda$ with absolute error distortion is given by
\begin{align}
R(D)=\left\{\begin{array}{ll}
-\log  (\lambda D), &0\leq D\leq \frac{1}{\lambda},\\
0,&D>\frac{1}{\lambda}.
\end{array}
\right.
\end{align}
Moreover, the optimal conditional distribution is
\begin{align}
f^*_{X|\hat{X}}(x|\hat{x})=\frac{1}{2D} e^{- |x-\hat{x}|/D},\quad x,\hat{x}\in\mathbb{R}.\label{fun:lap_optimal_conditional}
\end{align}
\end{lemma}

\begin{proof}
The proof is given by \cite{Cover:Laplace91}, where the noise in test channel is given by Laplacian with parameter $1/D$.
See also Appendix~\ref{sec:AppF}.
\end{proof}

\subsection{Expansion Coding}
By noting that Laplacian is two-sided exponential, the expansion of source and estimate over levels ranging from $-L_1$ to $L_2$ can be expressed as
\begin{align}
&X_i=X_i^{\text{sign}}\sum_{l=-L_1}^{L_2}2^lX_{i,l},\quad i=1,2,\ldots,n,\label{fun:expansion_source}\\
&\hat{X}_i=\hat{X}_i^{\text{sign}}\sum_{l=-L_1}^{L_2}2^l \hat{X}_{i,l},\quad i=1,2,\ldots,n, \label{fun:estimate_expansion}
\end{align}
where $X_i^{\text{sign}}$ and $\hat{X}_i^{\text{sign}}$ represent the sign of $X_i$ and $\hat{X}_i$ correspondingly, both random variables uniformly distributed from $\{-1,+1\}$.

In a manner similar to exponential source coding case, expansion reduces the original problem to coding for a set of independent binary sources. However, particularly for Laplacian case, we let $X^{\text{sign}}=\hat{X}^{\text{sign}}$, i.e. using 1 bit to perfectly recover the sign bit, and then for other levels, we formulate each as a binary source coding with Hamming distortion. In particular, for level $l$, we design a symmetric test channel from $\hat{X}_l$ to $X_l$, where the cross probability is given by
    \begin{equation}
    d_l=\frac{p_l-\hat{p}_l}{1-2\hat{p}_l}.\label{fun:d_l}
    \end{equation}
    Then, the achievable rate at level $l$ is given by
\begin{equation}
R_l=H(p_l)-H(d_l).
\end{equation}
We have the following result.
\begin{theorem}\label{thm:laplace_rate}
For Laplacian source $X$, expansion coding, where the estimate $\hat{X}$ is constructed as in \eqref{fun:estimate_expansion},
achieves the rate distortion pair $(R,D)$ with
\begin{equation}
R=1+\sum_{l=-L_1}^{L_2}\left[H(p_l)-H(d_l)\right],\label{fun:laplace_achievable_rate}
\end{equation}
for any $L_1,L_2>0$ and $d_l$ such that $\mathbb{E}[|X-\hat{X}|]\leq D$.
\end{theorem}
The absolute value error distortion $\mathbb{E}[|X-\hat{X}|]$ cannot be written as simple weighted sum of Hamming distortions from each level.
In fact, we have to use an induction method to characterize the complicated relation. Denote
\begin{align}
\mathcal{D}_{k}\triangleq \mathbb{E}\left[\left|\sum_{l=-L_1}^{k}2^l(X_l-\hat{X}_l)\right|\right],\label{fun:D_k}
\end{align}
for any $-L_1\leq k \leq L_2$, which represents the accumulative absolute value distortion up to level $k$.
\begin{itemize}
\item Initialization: at level $-L_1$,
\begin{equation}
\mathcal{D}_{-L_1}=2^{-L_1}d_{-L_1}. \nonumber
\end{equation}
\item Induction: for levels $-L_1+1\leq k\leq L_2$,

\begin{equation}
\hskip-3em \mathcal{D}_{k}=\mathcal{D}_{k-1}(1-d_{k})+2^{k}d_k+\frac{2^{k}d_{k}(1-2p_{k})}{1-2d_{k}}\sum_{l=-L_1}^{k-1}\frac{2^{l}d_{l}(1-2p_{l})}{1-2d_{l}}.\nonumber
\end{equation}
\end{itemize}

To this end, the expansion based coding scheme can be clearly expressed as an optimization problem with variables $\{d_{-L_1},\ldots,d_{L_2}\}$, but not convex. We have to step back to heuristically choose the value of $d_l$s in order to get a suboptimal result.
More precisely, for an aiming distortion $D$, we construct a set of distortions $d_l$ at each level,
\begin{eqnarray}
d_l=\frac{1}{1+e^{2^l/D}}.\label{eqn:dlLaplace}
\end{eqnarray}
Then by Theorem \ref{thm:laplace_rate} and iterative algorithm to calculate the real distortion $\mathcal{D}_{L_2}$, we are ready to claim that the rate distortion pair $(R^{(1)},D^{(1)})$ is achievable, where
\begin{align}
R^{(1)}&=1+\sum_{l=-L_1}^{L_2}\left[H(p_l)-H(d_l)\right],\\
D^{(1)}&=\mathcal{D}_{L_2}.
\end{align}

Evidently, this coding scheme may not behave well at high distortion region, since $R^{(1)}$ is at least 1. In the high-distortion regime, precisely compressing the sign bit seems inefficient. To this end, a time sharing scheme is utilized to reduce the gap in high distortion region. More precisely, for any $\alpha\in[0,1]$, we compress $\alpha$ fraction of source sequences into codeword 0, then the following rate distortion pair is found to be achievable:
\begin{align}
R^{(2)}&=(1-\alpha)R^{(1)},\\
D^{(2)}&=(1-\alpha)D^{(1)}+\alpha/\lambda.
\end{align}

The following theorem provides an upper bound on rate distortion gap of expansion coding scheme.
\begin{theorem}\label{thm:laplace_bound}
For any $D\in[0,1/\lambda]$, with a choice of $d_l$ in \eqref{eqn:dlLaplace} and $L_1,L_2>-\log(\lambda D)$, the achievable rate distortion pairs $(R^{(1)},D^{(1)})$ and $(R^{(2)},D^{(2)})$ obtained from expansion schemes above are within $1$ bit gap to Shannon rate distortion function, i.e.
\begin{align}
&R^{(1)}-R(D^{(1)})\leq 1,\nonumber\\
&R^{(2)}-R(D^{(2)})\leq 1.\nonumber
\end{align}
\end{theorem}
\begin{proof}
See Appendix~\ref{sec:AppG}.
\end{proof}

\subsection{Numerical Result}

We find that the expansion coding scheme is provably within 1 bit constant gap of the rate distortion function. Here, the calculation of $R^{(1)}$ is fairly tight, however, the upper bound on $D^{(1)}$ could prove to be loose, especially in low distortion region. Since the calculation of $D^{(1)}$ from $d_l$s is non-trivial, it is hard to characterize the extent to which the overall distortion is overestimated by the bound. Thus, we turn to numerical results to find   this gap to be  0.52 bits in the low distortion regime (shown in Fig.~\ref{fig:Rate}).

\begin{figure}[t]
 \centering
 \includegraphics[width=\columnwidth]{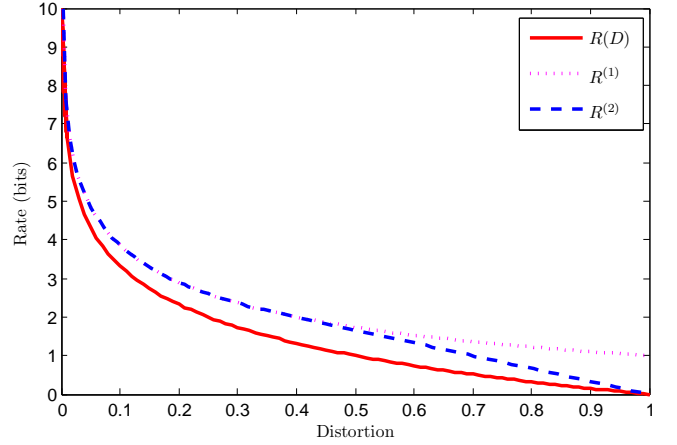}
 \caption{{\bf Achievable rate distortion pairs using expansion coding.} In this numerical result, we set $\lambda=1$. $R(D)$ (red-solid) is rate distortion limit; $R^{(1)}$ (purple-dotted) is achievable rate using expansion coding: and $R^{(2)}$ (blue-dashed) is achievable rate using expansion coding and time sharing.
}
\label{fig:Rate}
\end{figure}
%
%


\section{Discussion}

Expansion coding enables construction of ``good" lossy compression codes for exponential and Laplacian sources using  discrete-valued parallel source codes. Theoretical analysis and numerical results illustrate that expansion coding performs within a constant gap of the rate distortion limit, and therefore, approaches the rate distortion limit in ratio, in the low distortion regime.

One significant benefit from expansion coding is  coding complexity. As indicated in theoretical analysis, approximately $-2\log (\lambda D)$ number of levels are sufficient for the coding scheme  as presented and studied in the paper. Thus, by choosing ``good" low complexity  codes within each level (such as source coding with polar codes \cite{Arikan:Channel08}, \cite{Korada:Source10}), the overall complexity of the coding scheme can be easily characterized, resulting in a low-complexity net code for the original continuous-valued source coding problem.

Although the paper focuses primarily on binary expansion case, our results can be generalized to $q$-array expansion case, with similar performance guarantees. Moreover, we focus on exponential and Laplacian sources due to their decomposable property. As we can imagine, all decomposable distributions can be treated in a similar way to result in parallel problems. Even for indecomposable distributions, such as a Gaussian, the expansion coding scheme presents a means of developing low-complexity coding schemes for these types of sources.






\appendices
\section{Proof to Lemma~\ref{lem:exponential_expansion}}
\label{sec:AppA}

The ``if'' part follows by extending the one given in
\cite{Marsaglia:Random71}, which considers the expansion
of a truncated exponential random variable.
We show the result by calculating the
moment generating function of $B$.
Using the assumption that $\{B_l\}_{l\in\mathbb{Z}}$ are
mutually independent, we have
\begin{align}
M_B(t)  =\EE[e^{tB}]
        =\prod_{l=-\infty}^{\infty}\EE\left[e^{t2^l B_l}\right].\nonumber
\end{align}
Note that for any $l\in\mathbb{Z}$,
\begin{align}
\EE\left[e^{t2^l B_l}\right]=
\frac{e^{t2^l}}{1+e^{\lambda 2^l}}+\left(1-\frac{1}{1+e^{\lambda 2^l}}\right)
=\frac{1+e^{(t-\lambda ) 2^l}}{1+e^{-\lambda 2^l}}.\nonumber
\end{align}
Then, using the fact that for any constant $\alpha\in\mathbb{R}$,
\begin{align}
\prod_{l=0}^{n}(1+e^{\alpha 2^l})=\frac{1-e^{2^{n+1}\alpha}}{1-e^{\alpha}},\nonumber
\end{align}
we can obtain the following for $t<\lambda$,
\begin{align}
\prod_{l=0}^{\infty} \EE\left[e^{t2^l B_l}\right]
=\lim_{n\rightarrow\infty}\prod_{l=0}^{n}\frac{1+e^{(t-\lambda ) 2^l}}{1+e^{-\lambda 2^l}}
= \frac{1-e^{-\lambda}}{1-e^{t-\lambda}}.\label{eqn:part1}
\end{align}
And, similarly, for the negative part, we have
\begin{align}
\prod_{l=-n}^{-1}(1+e^{\alpha 2^l})=\frac{1-e^{\alpha}}{1-e^{\alpha2^{-n}}},\nonumber
\end{align}
which further implies that
\begin{align}
\prod_{i=-\infty}^{-1}\EE\left[e^{t2^i B_l}\right]
=&\lim_{n\rightarrow\infty}\frac{1-e^{t-\lambda}}{1-e^{(t-\lambda)2^{-n}}}\frac{1-e^{-\lambda2^{-n}}}{1-e^{-\lambda}}\nonumber\\
=&\frac{\lambda(1-e^{t-\lambda})}{(\lambda-t)(1-e^{-\lambda})}.
\label{eqn:part2}
\end{align}

Thus, finally for any $t<\lambda$, combining equations~(\ref{eqn:part1}) and~(\ref{eqn:part2}), we get
\begin{align}
M_B(t)=\frac{\lambda}{\lambda-t}.\nonumber
\end{align}
The observation that this is the moment generation function for an exponentially distributed random variable with parameter $\lambda$ concludes the proof.

The independence relationships between levels in ``only if'' part can be simply verified using memoryless property of exponential distribution. Here we just
need to show the parameter for Bernoulli random variable at each level. Observe that for any $l\in\mathbb{Z}$,
\begin{align}
\Prob\{B_l=1\}=\Prob\{B\in\cup_{k\in\mathbb{N}}[2^l(2k-1),2^l(2k))\}.\label{fun:exp}
\end{align}
Using CDF of exponential distribution, we obtain
\begin{align}
\Prob\{2^l(2k-1)\leq B<2^l(2k)\}&=e^{-\lambda2^l(2k-1)}-e^{-\lambda2^l(2k)}\nonumber\\
&=e^{-\lambda2^l(2k)}\left(e^{\lambda 2^l}-1\right).\nonumber
\end{align}
Putting this back to (\ref{fun:exp}) we have
\begin{align}
\Prob\{B_l=1\}=\sum_{k=1}^{\infty}e^{-\lambda2^l(2k)}\left(e^{\lambda 2^l}-1\right)=\frac{1}{e^{\lambda 2^l}+1}. \nonumber
\end{align}

\section{Proof of Lemma~\ref{lem:exp_rate_distortion}}
\label{sec:AppB}

Note that maximum entropy theorem tells us
the distribution maximizing differential entropy over all probability densities $f$ on support set $\mathbb{R}^+$ satisfying
$$\int_{0}^{\infty}f(x)xdx=0,$$
$$\int_{0}^{\infty}f(x)xdx=1/\lambda,$$
is exponential distribution with parameter $\lambda$. Based on this result, by noting $\mathbb{E}[d(X^n,\hat{X}^n)]\leq D$ ie equivalent to say $X\geq \hat{X}$ and $\mathbb{E}[X-\hat{X}]\leq D$, we have
\begin{align}
I(X;\hat{X})&=h(X)-h(X|\hat{X})\nonumber\\
            &=\log(\frac{e}{\lambda})-h(X-\hat{X}|\hat{X})\nonumber\\
            &\geq \log(\frac{e}{\lambda})-h(X-\hat{X})\nonumber\\
            &\geq \log(\frac{e}{\lambda})-\log(e\mathbb{E}[X-\hat{X}])\nonumber\\
            &\geq \log(\frac{e}{\lambda})-\log(eD)\nonumber\\
            &=-\log (\lambda D).\nonumber
\end{align}
Obviously, we need $X-\hat{X}$ to be exponentially distributed and independent with $\hat{X}$ as well. More specifically, we can design a test channel from $\hat{X}$ to $X$ with additive noise $Z=X-\hat{X}$ distributed as exponential with parameter $1/D$, which gives (\ref{fun:exp_optimal_conditional}).

\section{Proof of Theorem~\ref{thm:exp_rate_1}}
\label{sec:AppC}

Due to decomposability of exponential distribution, the levels after expansion are independent, hence, the achievable rate in this theorem is straightforward to get. On the other hand, for the calculation of distortion, we have
\begin{align}
D_1 &=\mathbb{E}[\sum_{l=-\infty}^{\infty}2^lX_l-\sum_{l=-L_1}^{L_2}2^l\hat{X}_l]\nonumber\\
    &=\sum_{l=-L_1}^{L_2}2^ld_l +\sum_{l=L_2+1}^{\infty} 2^lp_l+\sum_{l=-\infty}^{-L_1-1}2^lp_l\nonumber\\
    &\leq \sum_{l=-L_1}^{L_2}2^ld_l +\sum_{l=L_2+1}^{\infty} 2^{-l}/\lambda +\sum_{l=-\infty}^{-L_1-1}2^l\nonumber\\
    &\leq \sum_{l=-L_1}^{L_2}2^ld_l + 2^{-L_2}/\lambda+2^{-L_1},\nonumber
\end{align}
which gives the result of the theorem.

\section{Proof of Theorem~\ref{thm:exp_rate_2}}
\label{sec:AppD}

By the design of coding scheme, if all higher levels are decoded as equivalence, then they must be encoded with one-sided distortion. Recall that for $Z$-channel, we have
$$\text{Pr}\{X_l\neq \hat{X_l}\}=\text{Pr}\{X_l=1,\hat{X}=0\}=d_l.$$
Hence, due to independence of expanded levels,
$$q_l=\prod_{k=l+1}^{L_2}(1-d_k).$$
Then, at each level, the achievable rate is $R_l$ with probability $q_l$ and is $\bar{R}_l$ otherwise. Thus, we have the expression of $R_2$ given by the theorem.

\section{Proof of Theorem~\ref{thm:exp_bound}}
\label{sec:AppE}

Without loss of generality, we assume $\lambda=1$ for simplicity in the proof. The proof of the theorem is based on an asymptotic result from \cite{Ozan:Expansion12}, which is restated without proof as follow.
\begin{align}
&0<H(p_l)<3\log e\cdot 2^l \quad\text{for }l\geq 0,\label{fun:exponential_lemma_1}\\
&1>H(p_l)>1-\log e\cdot 2^l \quad \text{for }l\leq 0.\label{fun:exponential_lemma_2}
\end{align}
By noting that $d_l$ is also the parameter of expanded exponential distribution at level $l$, but with a different mean, we have
\begin{equation}
d_l=\frac{1}{1+e^{2^l/D}}=\frac{1}{1+e^{2^{l+\gamma}}}=p_{l+\gamma},
\end{equation}
where $\gamma \triangleq -\log D$.
This result shows values of $\{d_l\}$ are right-shifted version of $\{p_l\}$ by $\gamma$ positions. Using this fact, together with equation (\ref{fun:exponential_lemma_1}) and (\ref{fun:exponential_lemma_2}), we have
\begin{align}
    \sum_{l=-L_1}^{L_2}\left[H(p_l)-H(d_l)\right]
= &\sum_{l=-L_1}^{L_2}H(p_l)-\sum_{l=-L_1+\gamma}^{L_2+\gamma}H(p_l)\nonumber\\
=&\sum_{l=-L_1}^{-L_1+\gamma-1}H(p_l)-\sum_{l=L_2+1}^{L_2+\gamma}H(p_l)\nonumber\\
 \leq &\gamma.\label{fun:proof_part1}
\end{align}
Moreover, note that
\begin{align}
&H(d_l)-(1-p_l+d_l)H\left(\frac{d_l}{1-p_l+d_l}\right)\nonumber\\
=&(1-p_l)\log(1-p_l)-(1-d_l)\log(1-d_l)\nonumber\\
&\quad-(1-p_l+d_l)\log(1-p_l+d_l).\label{proof:exp_first_result}
\end{align}
We want to bound this for two cases:
\begin{itemize}
\item[1)]
For $l\leq-\gamma$, by using the fact that function $g(x)=x\log x$ is convex and increasing on $(0.5,1)$, we have
\begin{align}
&(1-p_l)\log (1-p_l)-(1-d_l)\log (1-d_l)\nonumber\\
=&g(1-p_l)-g(1-d_l)\nonumber\\
\leq &-(p_l-d_l)g'(1-p_l)\nonumber\\
= &-(p_l-d_l)\left[\log e+\log(1-p_l)\right].\label{proof:exp_part1}
\end{align}
Then, by noting that $\log(1-x)\geq -2x\log e$ for any $x\in(0,0.5)$, we get
\begin{align}
&-(1-p_l+d_l)\log (1-p_l+d_l)\nonumber\\
\leq & 2(1-p_l+d_l)(p_l-d_l)\log e\nonumber\\
\leq & 2(p_l-d_l)\log e.\label{proof:exp_part2}
\end{align}
Putting equation (\ref{proof:exp_part1}) and (\ref{proof:exp_part2}) back to (\ref{proof:exp_first_result}), we have
\begin{align}
    &H(d_l)-(1-p_l+d_l)H\left(\frac{d_l}{1-p_l+d_l}\right)\nonumber\\
\leq &(p_l-d_l)[\log e-\log(1-p_l)]\nonumber\\
\leq & 2\log e(p_l-d_l).\nonumber
\end{align}
Further by noting that for $l\leq -\gamma$,
\begin{equation}
d_l=p_{l+\gamma}=\frac{1}{1+e^{2^{l+\gamma}}}\geq \frac12- 2^{l+\gamma-1},\nonumber
\end{equation}
and combining with the fact that $p_l<1/2$, we have
\begin{equation}
\hskip-1em H(d_l)-(1-p_l+d_l)H\left(\frac{d_l}{1-p_l+d_l}\right)\leq \log e\cdot2^{l+\gamma}.\label{fun:proof_part21}
\end{equation}
\item[2)] On the other hand, for $l>-\gamma$, similarly we have
\begin{align}
&(1-p_l)\log (1-p_l)-(1-p_l+d_l)\log (1-p_l+d_l)\nonumber\\
\leq &-d_l\left[\log e+\log(1-p_l)\right],\label{proof:exp_part3}
\end{align}
and
\begin{align}
-(1-d_l)\log (1-d_l)\leq  2\log e\cdot d_l.\label{proof:exp_part4}
\end{align}
Putting equation (\ref{proof:exp_part3}) and (\ref{proof:exp_part4}) back to (\ref{proof:exp_first_result}), we have
\begin{align}
    &H(d_l)-(1-p_l+d_l)H\left(\frac{d_l}{1-p_l+d_l}\right)\nonumber\\
\leq &d_l[\log e-\log(1-p_l)]\nonumber\\
\leq &2\log e\cdot d_l.\nonumber
\end{align}
Note that for $l>-\gamma$,
\begin{equation}
d_l=p_{l+\gamma}=\frac{1}{1+e^{2^{l+\gamma}}}\leq 2^{-l-\gamma},\nonumber
\end{equation}
then,
\begin{equation}
\hskip-2em H(d_l)-(1-p_l+d_l)H\left(\frac{d_l}{1-p_l+d_l}\right)\leq \log e \cdot 2^{-l-\gamma+1}.\label{fun:proof_part22}
\end{equation}
\end{itemize}
Collecting all the pieces together, we have
\begin{align}
R^{(1)} &=\sum_{l=-L_1}^{L_2} \left[ H(p_l)-(1-p_l+d_l)H\left(\frac{d_l}{1-p_l+d_l}\right)\right]\nonumber\\
    &\overset{\text{(a)}}{\leq} \gamma +\sum_{l=-L_1}^{L_2} \left[ H(d_l)-(1-p_l+d_l)H\left(\frac{d_l}{1-p_l+d_l}\right)\right]\nonumber\\
    &\overset{\text{(b)}}{\leq} R(D)+\sum_{l=-L_1}^{-\gamma} 2^{l+\gamma-1}+\sum_{l=-\gamma+1}^{L_2}2^{-l-\gamma}\nonumber\\
    &\leq R(D)+4\log e,\label{fun:proof_part3}
\end{align}
where inequality (a) comes from \eqref{fun:proof_part1}, and (b) comes from \eqref{fun:proof_part21} and \eqref{fun:proof_part22}.
Finally, by noting the fact from Theorem~\ref{thm:exp_rate_1} that
$$D^{(1)}\leq D+2^{-L_2}+2^{-L_1},$$
and that rate distortion function is convex and decreasing, we have
$$R(D)\leq R(D^{(1)})+(2^{-L_2}+2^{-L_1})\log e/D \leq R(D^{(1)})+2\log e.$$
Relating this to \eqref{fun:proof_part3}, we have
$$R^{(1)}\leq R(D^{(1)})+6\log e,$$
which completes the proof for $R^{(1)}$ and $D^{(1)}$ by choosing $c=6\log e$.

For the other part of the theorem, $R^{(2)}$ and $D^{(2)}$, observe that
\begin{align}
H\left( \frac{d_l}{1-2p_l+2d_l} \right)&\geq H\left( \frac{d_l}{1-p_l+d_l} \right)\nonumber\\
                    &\geq (1-p_l+d_l)H\left( \frac{d_l}{1-p_l+d_l} \right).\nonumber
\end{align}
Hence, for any $-L_1\leq l\leq L_1$, we have $\bar{R}_l\leq R_l.$
Thus, by noting $R^{(2)}$ is a convex combination of $\bar{R}_l$ and $R_l$ at each level, we have $R^{(2)}\leq R^{(1)}$. Combing with the observation that $D^{(1)}=D^{(2)}$, we have $R^{(2)}\leq R(D^{(2)})+c$.

\section{Proof of Lemma~\ref{lem:laplace_rate_distortion}}
\label{sec:AppF}

Maximum entropy theorem tells us Laplace distribution with parameter $\lambda$ has the
maximum differential entropy $h(f)$ over all probability densities $f$ on support set $\mathbb{R}$ satisfying
$$\int_{-\infty}^{\infty}f(x)xdx=0,$$
$$\int_{-\infty}^{\infty}f(x)|x|dx=1/\lambda.$$
Based on this result, it is evident to note that
\begin{align}
I(X;\hat{X})&=h(X)-h(X|\hat{X})\nonumber\\
            &\geq \log(\frac{2e}{\lambda})-h(X-\hat{X})\nonumber\\
            &\geq \log(\frac{2e}{\lambda})-\log(2eD)\nonumber\\
            &=-\log (\lambda D),\nonumber
\end{align}
where we have used the fact that $\mathbb{E}[|X-\hat{X}|]\leq D$. Obviously, we need $X-\hat{X}$ to be Laplace distributed and independent with $\hat{X}$ as well. More specifically, we can design a test channel from $\hat{X}$ to $X$ with additive noise $Z=X-\hat{X}$ distributed as Laplace with parameter $1/D$, as shown in (\ref{fun:lap_optimal_conditional}).


\section{Proof to Theorem~\ref{thm:laplace_bound}}
\label{sec:AppG}

We assume $\lambda=1$ without loss of generality. From the proof of Theorem~\ref{thm:exp_bound}, we have already seen:
$$\sum_{l=-L_1}^{L_2}[H(p_l)-H(d_l)]\leq \gamma,$$
where $\gamma\triangleq-\log D$. Moreover, note that
\begin{align}
D^{(1)} &=\mathbb{E}\left[\left|\sum_{l=-L_1}^{L_2}2^l (X_l-\hat{X}_l)\right|\right]\leq D.\nonumber
\end{align}
Combining the pieces together, it is evident to see
\begin{align}
R^{(1)}-R(D^{(1)})  &=1+\sum_{l=-L_1}^{L_2}[H(p_l)-H(d_l)]+\log D^{(1)}\leq 1.\nonumber
\end{align}
On the other hand, $(R^{(2)},D^{(2)})$ is obtained by convex combination of $(R^{(1)},D^{(1)})$ and $(0,1/\lambda)$, thus, we also have $R^{(2)}\leq R(D^{(2)})+1$.



\end{document}